\documentclass[11pt]{article}
\usepackage{fullpage}

\usepackage{graphicx,epsfig}
\usepackage{array,float}
\usepackage{url}
\usepackage[usenames,dvipsnames]{xcolor}
\usepackage{amstext,amssymb,amsmath}
\usepackage{amsthm}
\usepackage{mathtools}
\usepackage{verbatim}
\usepackage{bm}
\usepackage{paralist}
\usepackage{ulem}
\usepackage[numbers]{natbib}
\usepackage{todonotes}
\usepackage[noend]{algorithmic}
\usepackage{enumitem}

\usepackage{caption}

\newtheorem{lem}{Lemma}[section]

\newtheorem{obs}[lem]{Observation}

\newtheorem{definition}{Definition}
\newtheorem{remark}{Remark}
\newtheorem{example}{Example}
\newtheorem{theorem}{Theorem}
\newtheorem{corollary}{Corollary}
\newtheorem{lemma}{Lemma}
\newtheorem{proposition}{Proposition}

\title{GSP --- The Cinderella of Mechanism Design\footnote{A preliminary work ``On the Truthfulness of GSP''~\cite{cavallo15} was presented at the 11th Ad Auctions Workshop and covered a number of results in this paper.}}

\author{
Christopher A. Wilkens\footnote{Yahoo Research, Sunnyvale, CA} \and
Ruggiero Cavallo\footnote{Yahoo Research, New York, NY} \and
Rad Niazadeh\footnote{Cornell University, Ithaca, NY. Part of this work was done while the author was an intern at Yahoo.}
}

\begin{document}

\maketitle

\begin{abstract}
Nearly fifteen years ago, Google unveiled the generalized second price (GSP) auction. By all theoretical accounts including their own~\cite{varian14}, this was the wrong auction --- the Vickrey-Clarke-Groves (VCG) auction would have been the proper choice --- yet GSP has succeeded spectacularly.

We give a deep justification for GSP's success: advertisers' preferences map to a model we call {\em value maximization}; they do not maximize profit as the standard theory would believe. For value maximizers, GSP is the truthful auction~\cite{aggarwal09}. Moreover, this implies an axiomatization of GSP --- it is an auction whose prices are truthful for value maximizers --- that can be applied much more broadly than the simple model for which GSP was originally designed. In particular, applying it to arbitrary single-parameter domains recovers the folklore definition of GSP. Through the lens of value maximization, GSP metamorphosizes into a powerful auction, sound in its principles and elegant in its simplicity.
\end{abstract}

\section{Introduction}  \label{sec:intro}

In February 2002, Google launched its pay-per-click AdWords auction with an innovation that would change the industry: while competitors used a first price model to sell ads on a search results page,\footnote{Overture.com, formerly GoTo.com and later bought by Yahoo, introduced the first sponsored search auctions in 1996 and was still the dominant player in 2002. Their auction had pay-your-bid pricing that has since been dubbed the generalized first price auction (GFP).} Google chose a second price paradigm~\cite{varian14}. The benefits were clear~\cite{edelman07}, and within 5 years, the search ads industry had standardized on Google's generalized second price (GSP) auction. There was one problem with GSP: it was not the proper generalization of Vickrey's second price auction, a title that instead belongs to the Vickrey-Clarke-Groves (VCG) mechanism. Google realized this mere months after launch, yet GSP was already humming along, too entrenched to be replaced~\cite{varian14}.

By all reasonable accounts GSP has thrived, much to the consternation of the research community. Our collective first response was simply ``it can't be too bad:'' revenue and welfare in reasonable equilibria are no worse than VCG~\cite{varian07,eos07}. Many studies of GSP followed --- a more recent justification falls out of GSP's simplicity: GSP pricing is independent of the likelihood that the user clicks on an ad,\footnote{Technically, the ``ad effect'' term in the click probability appears in both the optimization and pricing, but these can be thought of as bidder weights; analysis of GSP's equilibria does not require that they are accurate. This simplicity also makes implementation easier and is a practical reason GSP is common.} so GSP's equilibrium guarantees are immune to errors in click modeling~\cite{milgrom10,dutting16}. Yet, throughout these and many other results, the community's implicit perspective has remained constant: GSP is an ugly duckling that requires continual excuses.

We turn past analysis on its head and propose that GSP's runaway success follows from a much simpler explanation: {\em GSP is truthful}. Of course, it is well-known that VCG is truthful {\em for profit maximizers}. In reality, advertisers rarely optimize for profit, and, as we will discus, their goals and behavior are much closer to a model we call {\em value maximization}. In fact, it is quite easy to see that GSP is truthful for this class of bidders --- Aggarwal et al.~\cite{aggarwal09} note exactly this in the limiting case of a general preference model --- the challenge is to argue that advertisers' preferences {\em generally} fit our value maximizing model, which we do in two ways. First we argue based on advertisers' traditional objectives that the value maximizing model is the right minimal theoretical model. Second, we prove and empirically validate that any bidder with a moderate ROI constraint will behave equivalently to a value maximizer, making value maximization a robust choice for fitting advertiser preferences.

Moreover, our work suggests a fundamental axiomatic definition of GSP: it is the auction that maximizes the sum of advertisers' values and charges prices truthful for value maximizers. This is the first work that gives a rigorous foundation to (e.g.) generalize GSP when search ads outgrow their current model. In doing so, we justify the folklore definition that the GSP price is the minimum bid required to maintain the same allocation.

Through these results, we develop a new appreciation for GSP. It gains sound theoretical footing as a truthful auction, magnified by the fact that value maximizing preferences can be justified in multiple ways. This virtue is heightened by its elegant simplicity, making it a veritable Cinderella --- the wrong auction that turned out to be just right.

\subsection{Preferences of Online Advertisers}

Advertisers can be split into two categories: {\em brand} and {\em performance}. Brand advertisers target long-term growth and awareness. Their goal is to reach as many people as possible, generate as many clicks as possible, capture as much market share as possible, generate as much revenue as possible, and so on. In contrast, performance advertisers care about immediate results. They optimize for the sales, sign-ups, or other so-called {\em conversions} generated directly from their ads.

\paragraph{Preferences of Brand Advertisers}
Brand advertisers come to a marketplace with explicit objectives. They typically have a mandate to meet a specific business goal --- showing impressions to an audience, generating clicks, or maximizing revenue --- driven by long-term considerations instead of immediate profit. Thus cost, while important, merely enters their preferences as a constraint --- brand advertisers will have a budget and limits on what they are willing to pay on average for an impression/click/conversion, but otherwise directly optimize for their mandate.

\paragraph{Preferences of Performance Advetisers}
Performance advertisers optimize the immediate tradeoff between value and cost. Return on investment (ROI) has been the standard metric for measuring this tradeoff across all types of advertising for decades. ROI measures the ratio of the profit obtained (``return'') to the cost or price paid (``investment''), i.e., the density of profit in cost:
\[\texttt{ROI}=\frac{\texttt{Value-Price}}{\texttt{Price}}=\frac{\texttt{Profit}}{\texttt{Price}}\enspace\]
Being a density metric, unconstrained maximization of ROI is not sensible;\footnote{Unconstrained maximization of ROI would likely cause an advertiser to buy only the single cheapest impression or click available.} instead, advertisers come with an ROI constraint and optimize (e.g. maximize revenue and/or clicks). For example, the following story plays out regularly at Yahoo:
\begin{enumerate}
\item Advertiser X designates a small budget for testing a Yahoo advertising product.
\item Advertiser X measures ROI --- if X is happy with its ROI, it increases its budget hoping to {\em maintain the same ROI}; if it is unhappy, it withdraws.
\end{enumerate}
This behavior is also reflected in the standard industry tools: Google's AdWords campaign management tool buys as much advertising as possible while maintaining a target average CPC and budget (an average CPC constraint corresponds to an ROI constraint, while a marginal CPC constraint would be appropriate for profit maximization). Yahoo's display ad products also offer options with ``CPC goals'' and ``CPA goals'' that buy impressions while aiming to achieve a target average cost-per-click or cost-per-conversion respectively.

Performance advertisers' ROI-centric behavior requires comment because it {\em does not} maximize profit. Standard intuition suggests that advertisers who are actively trading off between value and cost should maximize profit, but it is easy to see that a ROI constraint fails to buy some profitable ads. One way to resolve this conflict is to conclude that the intuition is wrong, perhaps since companies often care about revenue and margins as much as profits. Another resolution is to conclude that ROI is a deeply-ingrained heuristic for maximizing the effectiveness of a fixed resource, either managing an advertising budget~\cite{borgs07,kitts04,szymanski08,zhou08,auerbach08} among auctions or between platforms (deciding how to bid between Yahoo and Google), or even managing resources within the advertiser's organization (allocating a budget between marketing and engineering). Regardless, however, performance advertisers behavior is rarely profit-maximizing.

\paragraph{Advertisers as Value Maximizers}

We have just seen that advertisers' objectives fit a simple paradigm that we call {\em value maximizing}: they maximize what they buy (value) subject to spending constraints. The simplest spending constraint is that the total price does not exceed the total value of goods purchased. This motivates our minimal theoretical model:\footnote{A similar model was discovered independently in concurrent
work by Fadaei and Bichler~\cite{fadaei16}. See related work.}
\begin{definition}[Informal]
A {\em simple value maximizer} maximizes value $v_i$ while maintaining budget balance, i.e. while keeping
payment $p_i\leq v_i$; when value is equal, a lower price is preferred.
\end{definition}

Surprisingly, this simple model even captures the auction-level preferences of advertisers who have budget and ROI constraints. Budget constraints can simply be ignored at the auction level because they are typically orders of magnitude larger than bids. In contrast, ROI constraints have a substantial effect at the auction level --- a value maximizer who requires ROI at least
$\gamma$ will require
$\frac{v_i(o)-p_i}{p_i}\geq\gamma$ and thus $p_i\leq \frac{v_i(o)}{1+\gamma}$ --- but this looks like a simple value maximizer with a different value:
\begin{obs}[Informal]
If an auction is truthful for simple value maximizers, value maximizers with ROI constraints will ``truthfully'' report $v_i'=\frac{v_i}{1+\gamma}$.
\end{obs}
Thus, if one na\"{i}vely runs a mechanism that is truthful for value maximizers, bidders with ROI constraints will simply report $v_i'=\frac{v_i}{1+\gamma}$, i.e. they report what they are willing to pay (the auctioneer does not need to know $\gamma$, and it does not need to be the same for all bidders). We will discuss these arguments more formally later.

\subsection{Related Work}

Aggarwal et al.~\cite{aggarwal09} design sponsored search (slot) auctions in a general matching model. In their model, each advertiser has a slot-specific value and a slot-specific maximum price; advertisers maximize profit subject to their maximum price constraints. In the limit where the values are large they call advertisers {\em maximum price} bidders, corresponding directly with our value maximizers, and observe that GSP is truthful for maximum price bidders. They treat maximum price bidders as a special corner of the market rather than as a general definition for all advertisers, and they do not build on their results to justify/interpret GSP.

A limited empirical literature studies the behavior of bidders. Varian~\cite{varian07} tests if bids can be explained as equilibria for some quasilinear bidders and finds that the errors are not large; this does not address our situation, though, because assuming bidders are value maximizers would result in zero error (everyone is simply telling the truth). Athey and Nekipelov~\cite{athey16} presuppose advertisers have broadly heterogenous objectives and model their behavior.

The idea that agents maximize value is common, though not in mechanism design. For example, the traditional studies of consumer choice and market equilibria begin with agents who pick the most-preferred bundle they can afford based on their endowment~\cite{mas-collel95}. In the context of mechanism design, only a limited literature studies non-quasilinear bidders. Concurrent work by Fadaei and Bichler introduce a related model of ``value bidder''~\cite{fadaei16}, then study a form of revenue maximization. Our model differs notably because we say bidders prefer lower prices when value is held constant, while Fadaei and Bichler imply bidders strictly prefer to spend their money. Alaei et al.~\cite{alaei11} show how Walrasian equilibria in non-quasilinear unit-demand settings can be leveraged to build ad auctions where estimation errors break quasilinearity. A few papers study general truthful mechanisms for non-quasilinear preferences~\cite{adachi13,morimoto15}. These papers develop axiomatic characterizations of a VCG-analog for multi-unit and unit-demand settings.

A few threads within the qasilinear advertising literature are also related. One line of work studies the effect of using ROI as a heuristic to set bids across different auctions~\cite{borgs07,kitts04,szymanski08,zhou08,auerbach08}. Another line of work uses ideas from behavioral economics to capture observed behavior~\cite{rong14}. Finally, a small thread studies GSP and the problem of generalizing it to new contexts. The canonical analyses of GSP in a quasilinear world are~\cite{eos07,varian07}. Attempts to generalize GSP and consider its performance outside the standard model include~\cite{cavallo14,hummel14,abrams07,aggarwal06}.

\section{Model and Preliminaries}  \label{sec:prelim}

A set of $n$ ads are to be placed on a page according to a configuration (outcome) $o\in\mathcal O$. For each configuration, the likelihood that the user clicks on ad $i$ ($i$'s allocation) is given by $x_i(o)$; when the user clicks, advertiser $i$ pays a price $p_i$.

Each advertiser has a private value per click $v_i$ and a public\footnote{We assume this is public for simplicity, but it is not necessary for the results in our paper.} complete and transitive preference $\preceq_i$ over outcome-price pairs $(o,p)$.

We will reference a few different models of preferences throughout the paper. First, we mention standard quasilinear preferences:
\begin{definition}
Preferences are {\em quasi-linear} if a bidder prefers to maximize the utility function $u_i=(v_i-p_i)x_i$, i.e.
\[(o,p_i)\succeq_i(o',p_i')\Leftrightarrow (v_i-p_i)x_i(o)\geq (v_i-p_i')x_i(o')\enspace.\]
\end{definition}

The majority of the paper will focus on value-maximizing preferences in various forms. The general value-maximizing preferences are defined for arbitrary constraints $\mathcal{C}$.

\begin{definition}
Preferences are {\em value-maximzing} subject to constraints $\mathcal{C}$ if they always prefer an outcome with higher value as long as none of its constraints are violated:
\[v_ix_i(o)>v_ix_i(o')~,~(o,p_i)\mbox{ satisfies }\mathcal{C}\quad\Rightarrow\quad (o,p_i)\succ_i(o',p_i')\enspace,\]
and when the value (allocation) is the same, a lower price is preferred:
\[x_i(o)=x_i(o')~,~p_i<p_i'~,~(o,p_i)\mbox{ satisfies }\mathcal{C}\Rightarrow (o,p_i)\succ_i(o',p_i')\enspace.\]
Moreover, any outcome is preferred to the one violating $\mathcal{C}$:
\[(o',p'_i)\mbox{ violates }\mathcal{C}\Rightarrow (o,p_i)\succeq_i(o,p_i')\enspace.\]
\end{definition}

We will primarily work with two specific types of monotone constraints:
\begin{definition}
A bidder has {\em simple value-maximizing preferences} if the bidder is value-maximizing subject to the constraint that $(o,p_i)$ must satisfy $p_i\leq v_i$. 
\end{definition}
\begin{definition}
A bidder has {\em value-maximizing preferences with an ROI constraint $\gamma$} if the bidder is value-maximizing with the constraint $\frac{v_i-p_i}{p_i}\geq\gamma$.
\end{definition}

\paragraph{Auctions and Truthfulness}

An auction asks bidders to report their value parameters and specifies, for bids $\mathbf{b}$ submitted by the bidders, an outcome $f(\mathbf{b})$ and payments $p_i(\mathbf{b})$. The function $f$ is known as the {\em social choice function} and implies an allocation $x_i(\mathbf{b})$.

A mechanism is {\em truthful} if reporting a bidder's true private information (e.g. bidding $b_i=v_i$) is a dominant strategy:

\begin{definition}
\label{def:dsic}
A mechanism $\mathcal M=(f,\mathbf{p})$ is {\em dominant strategy incentive compatible} (DSIC) if and only if 
\[\forall v_i, b,:\quad (f(b_{-i},v_i),p_i(b_{-i},v_i)) \succeq_i (f(b_{-i},b_i),p_i(b_{-i},b_i)) \]
\end{definition}

In this paper, we will have trouble in continuous bid spaces when truthful bidding would induce a tie; to solve this we introduce a slight weakening of DSIC that allows arbitrary behavior on a set of bids that ``never occur'':

\begin{definition}
\label{def:dsic-ae}
A mechanism $\mathcal M=(x,\mathbf{p})$ is {\em dominant strategy incentive compatible almost everywhere} (DSIC-AE) if, for any $i$ and bids $b_{-i}$, we have 
\[ (f(b_{-i},v_i),p_i(b_{-i},v_i)) \succeq_i (f(b_{-i},b_i),p_i(b_{-i},b_i)) \]
for all value functions $v_i$ and deviations $b_i$, except a set of $v_i$ with measure zero.
\end{definition}
We loosely use truthful to describe either DSIC or DSIC-AE.

\begin{remark}
Adding a minimum bid increment would eliminate the need for DSIC-AE (the next-highest bid is always well-defined) at a high cost in readability.
\end{remark}

\paragraph{Advertising Slot Auctions and GSP}

We will use with the standard
separable click-through-rate (CTR) framework used to study sponsored search
auctions. In the separable model, the $n$ ads compete for $m$ slots. There exist
$(\alpha_1,\ldots,\alpha_m)$ and $(\beta_1,\ldots,\beta_n)$ such that, when ad
$i$ is shown in slot $j$, the user clicks on it with probability:
\[\Pr[\mbox{click on ad $i$ when shown in slot $j$}]=\alpha_j\beta_i, \]
\if Consistent with the literature, we assume slots are labeled such that
$\alpha_1>\alpha_2\dots>\alpha_m$ (slot 1 is referred to as the top slot) and
bidders are labeled so that $\beta_1b_1\geq\beta_2b_2\dots\geq\beta_nb_n$. \fi

The standard GSP auction was originally defined in this setting:

\begin{definition}
The {\em generalized second price} auction (GSP) proceeds as follows:
\begin{enumerate}
\item Bidders are sorted by $\beta_ib_i$; the bidder ranked $j$ gets slot $j$.
\item A bidder is charged a price-per-click equal to the minimum bid required to keep the same slot.
\end{enumerate}
\end{definition}

\section{GSP is Truthful}  \label{sec:truthful}

\subsection{Value Maximizers and GSP}

We previously argued that advertisers' objectives are well-modeled by value maximization subject to a ROI constraint. We also hinted that this was equivalent to our simple value maximizer model, which we now formalize:
\begin{lemma}\label{lem:svm-roi}
If a deterministic auction is truthful for simple value maximizers, it will be a dominant strategy for a value maximizer with value $v_i$ and ROI constraint $\gamma$ to report what it is willing to pay, $v_i'=\frac{v_i}{1+\gamma}$.
\end{lemma}
Remark that we no-longer expect bidders to report $v_i$ and instead expect they will report what they are willing to pay. If we desired a true direct revelation mechanism we could ask bidders to report $\gamma$ as well, but in practice it is more natural to report what one is willing to pay than what something is truly worth.

With this lemma in hand, it remains to show that GSP is truthful for simple value maximizers. Surprisingly, this is nearly trivial, and was observed by Aggarwal et al.~\cite{aggarwal09} in the context of a general preference model.\footnote{The proof via~\cite{aggarwal09} observes that value maximizers with ROI constraints can be mapped to their maximum price bidder model.} We sketch a proof from first principles for intuition:
\begin{theorem}[Application of Aggarwal et al.~\cite{aggarwal09}]
GSP is truthful\footnote{Effectively DSIC-AE (see Definition~\ref{def:dsic-ae}).} for value maximizers with ROI constraints, that is, it is a dominant strategy for a bidder to report the maximum price she is willing to pay.
\end{theorem}
\begin{proof}[Sketch] (A rigorous generalization of this proof can be found in Theorem~\ref{thm:vm-sp-char}.) Note that the GSP price $p_i$ is the minimum bid $i$ could have submitted while maintaining the same rank in the auction.

First, assume that bidders are simple value maximizers. Fix a bidder $i$ and bids $b_{-i}$. By the taxation principle, we can think about bidder $i$ choosing among the slots at prices $p_{i,j}$. Let $V_{i,j}\subseteq\Re$ denote the set of value reports for which $i$ wins slot $j$ under a particular auction.

For a pricing to be truthful for a value maximizer, it should be that $i$ wins a particular slot $j$ if and only if (a) it is willing to pay $p_{i,j}$ for slot $j$ and (b) is not willing to pay $p_{i,j-1}$ for the slot immediately above it. This implies that (a) $p_j\leq\inf V_{i,j}$ and (b) $p_{j-1}\geq\sup V_{i,j-1}$ and thus
\[\sup V_{i,j+1}\leq p_{i,j}\leq\inf V_{i,j}\]
that is, $p_{i,j}$ is the threshold value (bid) at which $i$ moves from slot $j+1$ to slot $j$. This is precisely the GSP price.

Lemma~\ref{lem:svm-roi} extends this to value maximizers with an ROI constraint.
\end{proof}

\subsection{Robustness of Value Maximization}

Our next result aims to show that this minimal theoretical model is robust in a practical sense. Value maximization is intuitively an extreme behavior, so it would be unsurprising if small modeling changes dramatically changed our results. However, we find that the opposite happens in the presence of ROI constraints: a sufficient ROI constraint generally makes nearly any advertiser look like a value maximizer. An example is instructive:

\begin{example}
Suppose that a bidder has an ROI constraint of 1 and has choices between an outcome $o$ with value $v_ix_i(o)=1$ and an outcome $o'$ with value $v_ix_i(o')=10$. At any price $p'\leq\$9$, a profit-maximizing bidder will prefer $o'$ to $o$; however, since bidder $i$ has an ROI constraint of 1, she will only consider outcome $o'$ when $p'\leq\$5$. As a result, any time $o'$ is cheap enough for her to consider it, she will always prefer it to outcome $o$ --- in this example, {\em bidder $i$ is effectively a simple value maximizer with value $v_i'=\frac12v_i$.}
\end{example}

More generally, when outcomes in an auction lead to dramatically different allocations, a bidder must pay a very high price before a lesser outcome (lower value) would become preferable. In such cases, a mild ROI constraint --- which caps the price an advertiser is willing to pay --- will push a bidder towards simple value maximizing behavior.

To formalize this, we define a broad class of preference relations that captures everything from quasilinear to value-maximizing behaviors:
\begin{definition}
A preference relation $\prec_i$ is {\em super-quasilinear} if an advertiser who prefers a higher-value option under quasi-linear preferences also always prefers this option under $\prec_i$, i.e.
\[v_ix_i(o)\geq v_ix_i(o') \mbox{ and } (v_i-p_i)x_i(o)\geq (v_i-p_i')x_i(o')\Leftrightarrow (o,p_i)\succeq_i(o',p_i')\enspace.\]
\end{definition}

For super-quasilinear bidders we derive the following condition under which all preferences look value-maximizing:

\begin{theorem}\label{thm:robust}
When bidders have super-quasi-linear preferences $\prec_i$, an ROI constraint $\gamma$, and
\[v_ix_i(o)\frac{\gamma}{\gamma+1}\geq v_ix_i(o')\mbox{ for all }v_i,o,o'\mbox{ where }x_i(o)>x_i(o')\enspace,\]
no bidder wishes to lie in an auction that is truthful, is individually rational for value maximizers, and has no positive transfers.\footnote{No positive transfers means that the auctioneer never pays the bidders; individual rationality says that every bidder is at least as happy as if she had not participated at all.}
\end{theorem}
\begin{proof} Let $o$ and $p_i$ be the outcome and price that the bidder sees if it reports truthfully. Let $o'$ and $p_i'$ be any outcome and price that the bidder can achieve by lying.

First, suppose that $x_i(o')<x_i(o)$, i.e. bidder $i$ is lying to achieve an outcome with a lower value. Then we can use individual rationality
\[p_i\leq\frac{v_i}{\gamma+1}\]
and we know from no-positive-transfers that $p_i'\geq0$. Thus, by the conditions of the theorem, we know that
\[(v_i-p_i)x_i(o)\geq v_i\frac{\gamma}{\gamma+1}x_i(o)>v_ix_i(o')\geq (v_i-p_i')x_i(o')\enspace.\]
In words, a profit-maximizing bidder prefers truthful reporting at the maximum individually-rational price to lying, even if $o'$ happened for free. Since bidder $i$ has super-profit-maximizing preferences and $v_i(o)>v_i(o')$, it follows $(o,p_i)\succ_i(o',p_i')$.

Next, suppose that $x_i(o')=x_i(o)$. We know that the mechanism is truthful for value maximizers, so it must be that $p_i\leq p_i'$, otherwise a value maximizer would lie to achieve $o'$, $p_i'$. Since any super-quasilinear bider will prefer a lower price for the same allocation, we can conclude that $(o,p_i)\succeq_i(o', p_i')$.

Finally, suppose that $x_i(o')>x_i(o)$. Since the auction is truthful for value maximizers, we know that $p_i'>v_i$, otherwise a value maximizer would lie to achieve $o'$, $p_i'$. By individual rationality we know that any super-quasilinear bidder prefers $(0,0)$ an outcome $o'$ at price $p'>v$. Thus, $(o,p_i)\succeq_i(0,0)\succ_i(o',p_i')$, so bidder $i$ will therefore prefer to tell the truth.
\end{proof}

\subsubsection{Empirical Robustness}

To empirically evaluate Theorem~\ref{thm:robust}, we look at Yahoo marketplace data and ask {\em for what ROI constraint can we safely conclude that any super-quasilinear bidder would behave like a value maximizer?} If we directly apply the theorem to an the separable model we get the following conservative corollary:
\begin{corollary}[of Theorem~\ref{thm:robust}]
If $\alpha_i\geq\frac{\gamma}{\gamma+1}\alpha_{i+1}$ for all $i$ in a standard sponsored search auction, then no bidder with super-quasi-linear preferences and a ROI constraint of $\gamma$ wishes to lie under GSP pricing.
\end{corollary}
This claim gives conditions under which no bidder has an incentive to lie, independent of other bidders' bids; however, reality may be even stricter. Given bids, we can ask a simpler question of whether any bidder could possibly prefer to lie under current marketplace conditions:
\begin{lemma}\label{lem:native-robust}
When no two bidders have the same score $\beta t$, and
\[\forall i,\; \frac{\alpha_i}{\alpha_i-\alpha_{i+1}}-\frac{\alpha_{i+1}}{\alpha_i-\alpha_{i+1}}\frac{\beta_{i+2}b_{i+2}}{\beta_{i+1}b_{i+1}}<\gamma\enspace,\]
no bidder with super-quasi-linear preferences and a ROI constraint of $\gamma$ wishes to lie under GSP pricing.
\end{lemma}
The proof (omitted) is similar to Theorem~\ref{thm:robust} and a version can be found in an earlier version of this work~\cite{cavallo15}.

We tested this lemma empirically by looking at bid data for the slot
auctions on Yahoo's homepage stream. We took a dataset consisting of over one
hundred thousand auctions from a brief period of time within a single day.

Our results are striking --- if bidders require an ROI of 1, then 80\% of auctions would be such that no bidder can benefit by lying under GSP pricing. This strongly suggests that GSP may in fact be the appropriate auction for this setting. See Figure~\ref{fig:emp}.

\begin{figure}[h!]
\centering  
\hspace{-0mm}
\epsfxsize=3.5in \epsfbox{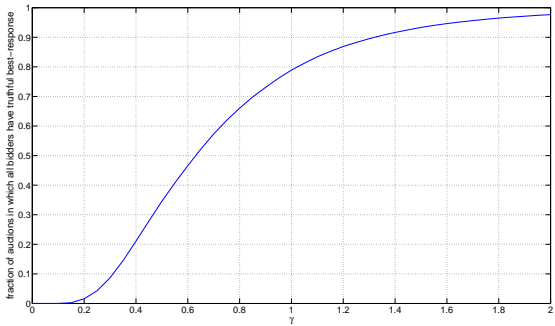}
\caption{\label{fig:emp} Illustration of the proportion of auctions in which truthtelling is a best-response for every advertiser, given the bids of all others, assuming all bidders have a ROI constraint of $\gamma$. At $\gamma=1$, 80\% of auctions are such that nobody should lie; as $\gamma$ approaches 2, virtually all auctions satisfy Lemma~\ref{lem:native-robust}. This is derived from a dataset of auctions from Yahoo's homepage stream, in which slot advertisements are interspersed in a stream of rich content links.}
\vspace{-0mm}
\end{figure}

\section{GSP Defined and Generalized}  \label{sec:general}

A significant byproduct of our results in Section~\ref{sec:truthful} is that we now have a theoretically justified axiomatic definition of GSP:
\begin{proposition}
\label{prop:ggsp}
The generalized second price auction (GSP) is the auction that maximizes expected bidder value and charges prices truthful for simple value maximizers.
\end{proposition}
This definition is important because it allows us to soundly extend GSP to new contexts. A small existing literature has struggled to generalize GSP in the context of quasilinear bidders, generally finding equilibria may not exist and that even when they do, they may not be efficient~\cite{deng10,cavallo14,bachrach14}. Our work suggests that these negative results are because GSP's strength was misunderstood.

Generalizing GSP is an important problem because ads are becoming progressively more complex, and with them the algorithms required to optimize their placement and features. The traditional separable model is a poor fit for these settings, so GSP does not immediately apply and the industry has typically moved to a VCG auction due to its purported truthfulness and elegant implementation. Facebook and Google's Contextual Ads~\cite{varian14} adopted VCG for precisely these reasons. However, those implementations of VCG were straightforward because they were largely from scratch; VCG is a more problematic prescription for established marketplaces like sponsored search that are currently based on GSP. Advertisers are accustomed to GSP prices, so any switch will force advertisers to revise their expectations. Moreover, GSP prices are higher than their VCG counterparts, so to the degree that advertisers do not react, there may be a significant revenue loss for the advertising platform.  A small research thread has explored ways to transition smoothly from GSP to VCG (e.g.~\cite{bachrach15}), again assuming bidders maximize profit.

\subsection{Single Parameter Generalizations}

To generalize GSP we will characterize the allocations that can be implemented truthfully for value maximizers. In quasilinear settings, it is known that any monotone allocation rule is incentive compatible~\cite{myerson81, archer01}; we find that the same is true for value maximizers. This auction has very simple prices: the prices is the minimum values required to maintain the same allocation $x$:

\begin{theorem}\label{thm:vm-sp-char}
For simple value maximizers in a single parameter domain, a mechanism is truthful (DSIC-AE) if and only if it is monotone and the price is the minimum value required to get the same allocation:
\begin{itemize}
\item {\em (monotonicity)} for any $v_{-i}$, $x_i(z, v_{-i})> x_i(v)$ if and only if $z>v_i$ almost everywhere over $v_i$
\item {\em (pricing)} $p_i(v)=\inf_{z|x_i(z, v_{-i})=x_i(v)}z$
\end{itemize}
\end{theorem}
Before we prove our characterization we note an important corollary:
\begin{corollary}
For any monotone $x_i$, the mechanism that charges $p_i(v)=\inf_{z|x_i(z, v_{-i})=x_i(v)}z$ is DSIC-AE.
\end{corollary}
This corollary follows because the $\inf$ only differs from the $\min$ where $x$ is discontinuous, and a monotone function is continuous almost everywhere (discontinuities are countable and therefore have measure zero).

\begin{proof}[of Theorem~\ref{thm:vm-sp-char}.] Fix other bidders' values $v_{-i}$ and drop them for clarity.

{\em Necessity.} Fix a type $v_i$ and define
\[\mathcal{Z}^<=\{z|x_i(z)<x_i(v_i)\}\quad\mbox{and}\quad \mathcal{Z}^==\{z|x_i(z)=x_i(v_i)\}\]
to be the types that get a strictly smaller allocation than $v_i$ and an equal allocation respectively. A value maximizer with single parameter type $v_i$ will choose a bid $b_i$ that maximizes $x_i(b_i)$ subject to $p_i(b_i)\leq v_i$. Thus, for any $v_i$, $p_i(v_i)$ must be high enough that bidders who get a smaller allocation do not want to lie, i.e.
\[p_i(v_i)\geq \underline p_i\triangleq\sup \mathcal{Z}^<\enspace.\]
Similarly, all types who exactly get $x_i(v_i)$ must pay the same price, and it must be that any bidder who gets $x_i(v_i)$ must be willing to pay for it, so
\[p_i(v_i)\leq\overline p_i=\inf \mathcal{Z}^=\enspace.\]
Since $\sup\mathcal{Z}^<\geq\inf\mathcal{Z}^=$, it must be that $\underline p_i\geq\overline p_i$, and therefore to satisfy $\underline p_i\leq p_i(v_i)\leq \overline p_i$ a mechanism will be truthful if and only if $\underline p_i=\overline p_i=p_i(v_i)$ --- this happens if and only if $x$ is nondecreasing and is the price defined in the theorem.

{\em Sufficiency.} Note that for DSIC-AE we only need to prove sufficiency almost everywhere. We must consider three deviations: (1) a bidder underbids for a smaller allocation, (2) a bidder deviates for the same allocation at a lower price, and (3) a bidder overbids for a larger allocation. Remark that deviation (1) will not happen because a value maximizer will always want the larger allocation (since $p_i(v_i)\leq v_i$ by definition of $p_i$), and that  (2) cannot happen because all types who get the same allocation pay the same price.

It remains to show that bidders almost never have an incentive to raise a bid (3). Deviation (3) will only be beneficial if there exists a type $z>v_i$ such that $x_i(z)>x_i(v_i)$ but $p_i(z)<v_i$. By definition of $p_i$, there exists a type $\underline b_i(z)\geq v_i$ such that
\[p_i(z)=\underline b_i(z)\enspace.\]
As long as $\underline b_i(z)>v_i$ this deviation cannot be beneficial. Unfortunately, when $b_i(z)=v_i$ then deviation (3) will be beneficial; however, this can only happen at a discontinuity in $x$. Since $x$ is monotone, discontinuities can only happen rarely (the set of discontinuities is countable and therefore has measure 0) and we get truthfulness in our almost-everywhere sense (DSIC-AE).
\end{proof}

\subsection{Substantiating GSP's Folk Definition}

While we have given the first rigorous axiomatic definition of GSP, the following folk description is well-known:
\begin{proposition}[Folklore]
GSP is the auction that maximizes expected bidder value and charges the minimum bid required to keep the same allocation.
\end{proposition}
Until now, this definition was merely a pleasant-sounding heuristic for general auction settings; however, our work gives this definition theoretical teeth:
\begin{obs}
The folklore definition of GSP is equivalent to the truthful auction for value maximizers defined in Theorem~\ref{thm:vm-sp-char}.
\end{obs}
By consequence, this folklore definition is also equivalent to Proposition~\ref{prop:ggsp}, closing the loop on defining GSP.

\section{Conclusion}  \label{sec:conclusion}

Our work paints GSP as a sound, practical, and elegant auction, nearly perfectly suited to the online advertising auctions it serves. Yet, it is undeniable that rejecting profit as an objective goes against the grain. Lest one finds earlier arguments unconvincing, we offer one concluding observation: {\em advertisers measure average costs, not marginal costs.} Every dashboard, spreadsheet, and presentation we have seen in practice starts by measuring performance in terms of the average cost per click, average value, and so on. Comparing averages measures ROI, not profit.

In light of our work, GSP almost appears as a happy coincidence. It wasn't designed to be truthful --- it wasn't even designed by auction experts as far as we know --- but it turned out to be a smashing success. Yet, perhaps GSP's original designer(s) had instinct that naturally led them to GSP, or perhaps in their na\"{i}vety they thought more like advertisers an auction expert would. This leaves a quiet, nagging question: what of that instinct might we still be missing?

\section{Acknolwedgements}

We would like to thank Prabhakar Krishnamurthy, who contributed to a precursor to this paper, as well as Sam Taggart, Mukund Sundararajan, and many anonymous reviewers for their helpful feedback.

\bibliographystyle{plain}
\bibliography{vm}

\end{document}